\documentclass[preprint,12pt]{elsarticle}
\journal{Theoretical Computer Science}

\usepackage{epsf}
\usepackage[dvips]{epsfig,graphicx}
\usepackage{subfigure}

\usepackage{latexsym}
\usepackage{textcomp}

\usepackage{amsmath}
\usepackage{amssymb}
\usepackage{mathrsfs}
\usepackage{dsfont}
\usepackage{tabularx}
\usepackage{multirow}
\usepackage{booktabs}
\usepackage{array}
\usepackage{algorithm}
\usepackage{algorithmic}
\floatname{algorithm}{Algorithm}

\newtheorem{theorem}{Theorem}

\newtheorem{corollary}[theorem]{Corollary}

\newtheorem{proposition}[theorem]{Proposition}

\newenvironment{proof}[1][Proof.\ ]
    {\noindent\textbf{#1}}
    {\hspace{0.4cm}\rule{.2cm}{.2cm}\vspace{\baselineskip}}

\newcolumntype{C}[1]{>{\centering\let\newline\\\arraybackslash\hspace{0pt}}m{#1}}

\def\P{\mathscr{P}}

\def\C{\mathcal{C}}

\def\Py{\mathrm{P}}

\def\NPC{\mathrm{NPc}}

\def\CONP{\mathrm{CoNP}}

\def\NP{\mathrm{NP}}

\begin{document}

\begin{frontmatter}

\title{A general framework for path convexities}

\cortext[cor]{Corresponding author.}
\author[uff,cefet]{Jo\~ao Vinicius C. Thompson}\ead{joao.thompson@cefet-rj.br}
\author[uff]{Loana T. Nogueira}\ead{loana@ic.uff.br}
\author[uff]{F\'abio Protti\corref{cor}}\ead{fabio@ic.uff.br}
\author[uff]{Raquel S. F. Bravo}\ead{raquel@ic.uff.br}
\author[ufrj]{Mitre C. Dourado}\ead{mitre@nce.ufrj.br}
\author[uff]{U\'everton S. Souza}\ead{usouza@ic.uf.br}

\address[uff]{Instituto de Computa\c c\~ao \\
Universidade Federal Fluminense \\
Niter\'oi - RJ, Brazil}

\address[cefet]{
Centro Federal de Educa\c c\~ao Tecnol\'ogica Celso Suckow da Fonseca\\
Campus Petr\'opolis\\
Petr\'opolis - RJ, Brazil}

\address[ufrj]{Depto de Ci\^encia da Computa\c c\~ao - Instituto de Matem\'atica \\
Universidade Federal do Rio de Janeiro \\
Rio de Janeiro - RJ, Brazil}

\begin{abstract}
In this work we deal with the so-called {\em path convexities}, defined over special collections of paths. For example, the collection of the shortest paths in a graph is associated with the well-known {\em geodesic convexity}, while the collection of the induced paths is associated with the {\em monophonic convexity}\,; and there are many other examples. Besides reviewing the path convexities in the literature, we propose a general path convexity framework, of which most existing path convexities can be viewed as particular cases. Some benefits of the proposed framework are the systematization of the algorithmic study of related problems and the possibility of defining new convexities not yet investigated.
\end{abstract}


\end{frontmatter}


\section{Introduction}

A {\em finite convexity space} is a pair $(V,\C)$ consisting of a finite set $V$ and a family $\C$ of subsets of $V$ such that $\emptyset\in\C$, $V\in\C$, and $\C$ is closed under intersection. Members of $\C$ are called {\em convex sets}.

Let $\P$ be a collection of paths of a graph $G$, and let $I_{\P}:2^{V(G)}\rightarrow 2^{V(G)}$ be a function (called {\em interval function}) such that
$$I_{\P}(S) = S\,\cup\,\{z\not\in S\mid\exists \ u,v\in S \ \mbox{such that} \ z \ \mbox{lies in an} \ uv\mbox{-path} \ P\in\P\}.$$

Distinct choices of $\P$ lead to interval functions of quite different behavior. Such functions, in turn, are naturally associated with special convexity spaces (the so-called {\em path convexities}). For instance, if $\P$ contains precisely all the shortest paths in a graph then the corresponding interval function is naturally associated with the well-known {\em geodesic convexity}\,; if $\P$ is the collection of induced paths then the corresponding interval function is associated with the {\em monophonic convexity}\,; and there are many other examples in the literature.

In this work we propose a general path convexity framework, of which most path convexities in the literature can be viewed as particular cases. Some benefits of the proposed framework are the systematization of the algorithmic study of related problems and the possibility of defining new path convexities not yet investigated.

Our contributions are concentrated mainly in Section 3, where we describe in detail our framework. The idea is to control the length of the paths in $\P$, as well as the types of chords allowed to exist in such paths. Such control can be done by means of four matrices that specify, for each pair $(u,v)$ of vertices, the minimum/maximum length and minimum/maximum chord length in all $uv$-paths of $\P$. We prove hardness results for the more general approach, where the matrices are part of the input of the related computational problems. We also describe some polynomial cases by restricting the usage of such matrices, including linear-time methods for bounded treewidth graphs. In addition, we show how to define most existing path convexities in the literature within the proposed framework. In Section 4 we provide examples of new interesting convexities and discuss future algorithmic developments.

\section{Preliminaries}

In this section we first provide all the necessary background. Next, we briefly review the main path convexities in the literature and list six fundamental computational problems in graph convexity that will be considered in this work. Finally, we prove two useful propositions.

All graphs are finite, simple, nonempty, and connected. Let $G$ denote a graph with $n$ vertices and $m$ edges. The {\em length} of a path $P$ in $G$, denoted by $|P|$, is its number of edges. A path $P$ in $G$ with endpoints $u$ and $v$ is an $uv${\em -path}. An $uv$-path $P$ in $G$ is {\em shortest} if there is no $uv$-path $P'$ in $G$ such that $|P'|<|P|$. If an $uv$-path $P$ is shortest then $|P|$ is the {\em distance} between $u$ and $v$ in $G$, and we write $|P|={\mathit dist}_G(u,v)$. A {\em chord of length} $l\geq 2$ in a path $P=(v_0,v_1,\ldots,v_{|P|})$ is an edge $v_iv_j\in E(G)$ such that $i,j\in\{0,\ldots,|P|\}$ and $|i-j|=l\geq 2$.


Let $\P$ be a collection of paths of a graph $G$, and let $I_{\P}:2^{V(G)}\rightarrow 2^{V(G)}$ be the interval function associated with $\P$, i.e.,
\begin{equation}\label{def-of-interval}
I_{\P}(S) = S\,\cup\,\{z\not\in S\mid\exists \ u,v\in S \ \mbox{such that} \ z \ \mbox{lies in an} \ uv\mbox{-path} \ P\in\P\}.
\end{equation}

Define $\C_{\P}$ as the family of subsets of $V(G)$ such that $S\in\C_{\P}$ if and only if $I_{\P}(S)=S$. Then it is easy to see that $(V(G),\C_{\P})$ is a finite convexity space, whose convex sets are precisely the fixed points of $I_{\P}$.

\begin{proposition}\label{prop:convexity}{\em \cite{V93}}
$(V(G),\C_{\P})$ is a finite convexity space.
\end{proposition}


In order to ease the notation, we omit the subscript ${\P}$ whenever it is clear from the context.

\subsection{Path convexities in the literature}

By varying the choice of the collection $\P$, interval functions of different behavior can be defined using Equation (\ref{def-of-interval}). The convexity spaces associated with such functions are called {\em path convexities}.

In Table~\ref{convexities} we list the main path convexities that appear in the literature. In the table, each convexity is defined by the collection of paths $\P$ considered.

\begin{table}[htb]
\centering
\small
\begin{tabular}{ |l|c|c|c|c|c| }
\toprule
convexity name & collection of paths $\P$ considered\\ \hline
geodesic \cite{HN81,p113,p148}          & shortest paths \\
monophonic \cite{p41,D88,p100}          & induced paths \\
$g^3$ \cite{p150}                       & shortest paths of length at least three \\
$m^3$ \cite{p28,p88}                    & induced paths of length at least three \\
$g_k$ \cite{p101}                       & shortest paths of length at most $k$ \\
$P_3$ \cite{p37,p85,p153}               & paths of length two \\
$P_3^*$ \cite{A13}                      & induced paths of length two \\
triangle-path \cite{p40,p43,p45}        & paths allowing only chords of length two \\
total  \cite{p76}                       & paths allowing only chords of length at least three \\
detour  \cite{p46,p68,p70}              & longest paths \\
all-path \cite{p42,p107,p165}           & all paths \\
\bottomrule
\end{tabular}
\caption{Some path convexities studied in the literature.}\label{convexities}
\end{table}

\subsection{Computational problems}\label{problems}

In this work we focus on six computational problems that are usually studied in the field of convexity in graphs. The list, of course, is not complete and other important problems could also be considered.

We need some additional definitions. Let $S\subseteq V(G)$. If $I(S) = V(G)$ then $S$ is an {\em interval set}. The {\em convex hull} $H(S)$ of $S$ is the smallest convex set containing $S$. Write $I^0(S)=S$ and define $I^{i+1}(S)=I(I^i(S))$ for $i\geq 0$. Note that $I(S)=I^1(S)$ and there exists an index $i$ for which $H(S)=I^i(S)$. If $H(S) = V(G)$ then $S$ is a {\em hull set}. The {\em convexity number} $c(G)$ of $G$ is the size of a maximum convex set $S\neq V(G)$. The {\em interval number} $i(G)$ of $G$ is the size of a smallest interval set of $G$. The {\em hull number} $h(G)$ of $G$ is the size of a smallest hull set of $G$. Now we are in position to state the six problems dealt with in this work:

\noindent {\sc Convex Set - CS}\\
Input: A graph $G$ and a set $S\subseteq V(G)$.\\
Question: Is $S$ convex?

\noindent {\sc Interval Determination - ID}\\
Input: A graph $G$, a set $S\subseteq V(G)$, and a vertex $z\in V(G)$.\\
Question: Does $z$ belong to $I(S)$?

\noindent {\sc Convex Hull Determination - CHD}\\
Input: A graph $G$, a set $S\subseteq V(G)$, and a vertex $z\in V(G)$.\\
Question: Does $z$ belong to $H(S)$?



\noindent {\sc Convexity Number - CN}\\
Input: A graph $G$ and a positive integer $r$.\\
Question: Is $c(G)\geq r$?

\noindent {\sc Interval Number - IN}\\
Input: A graph $G$ and a positive integer $r$.\\
Question: Is $i(G)\leq r$?

\noindent {\sc Hull Number - HN}\\
Input: A graph $G$ and a positive integer $r$.\\
Question: Is $h(G)\leq r$?

\subsection{Existing complexity results}

The table below shows the complexity of the six problems listed in the preceding subsection for some convexity spaces. All the entries of the table correspond to results found in the literature, or to trivial results (indicated by `[t]').

\def \mc #1{\multicolumn{2}{|c|}{#1}}

\begin{table}[H]
\centering
\small
\begin{tabular}{|l|cc|cc|cc|cc|cc|}
\toprule
         &\mc{geodesic}       &\mc{monophonic}     &\mc{$P_3$}          &\mc{$P_3^*$}        &\mc{triangle-path}\\\hline
{\sc cs} & $\Py$  &[t]        & $\Py$  &\cite{D10} & $\Py$  &[t]        & $\Py$  &\cite{A13} & $\Py$  &\cite{D15}\\
{\sc id} & $\Py$  &[t]        & $\NPC$ &\cite{D10} & $\Py$  &[t]        & $\Py$  &\cite{A13} & $\NPC$ &\cite{D15}\\
{\sc chd}& $\Py$  &[t]        & $\Py$  &\cite{D10} & $\Py$  &[t]        & $\Py$  &\cite{A13} & $\Py$  &\cite{D15}\\
{\sc cn} & $\NPC$ &\cite{G03} & $\NPC$ &\cite{D10} & $\NPC$ &\cite{C09} & $\NPC$ &\cite{A13} & $\Py$  &\cite{D15}\\
{\sc in} & $\NPC$ &\cite{A02} & $\NPC$ &\cite{D10} & $\NPC$ &\cite{C11} & $\NPC$ &\cite{A13} & $\NPC$ &\cite{D15}\\
{\sc hn} & $\NPC$ &\cite{D09} & $\Py$  &\cite{D10} & $\NPC$ &\cite{C09} & $\NPC$ &\cite{A13} & $\Py$  &\cite{D15}\\
\bottomrule
\end{tabular}
\caption{Problems vs Convexities: complexity results.}\label{conv-vs-prob}
\end{table}

\subsection{Two useful facts}

The next two propositions are useful. They say that if {\sc Interval Determination} or {\sc Convex Set} can be solved in polynomial time for some convexity space then some other problems listed in Section~\ref{problems} can also be solved in polynomial time, for the same convexity space.

\begin{proposition}
\label{prop:idpo}
Let $(V(G),\C)$ be any convexity space. If {\sc Interval Determination} can be solved in polynomial time for $(V(G),\C)$ then {\sc Convex Set} and {\sc Convex Hull Determination} can also be solved in polynomial time for $(V(G),\C)$.
\end{proposition}

\begin{proof}

Let $S\subseteq V(G)$. Since {\sc Interval Determination} is in $\Py$ for $(V(G),\C)$, $I(S)$ can be computed in polynomial time. Let $i$ be the smallest index such $I^{i+1}(S) = I^{i}(S)$. Note that determining such an index $i$ as well as $I^{i}(S)$ can also be done in polynomial time, since $i=O(n)$. Therefore:

  \begin{itemize}
    \item If $I(S) = S$ then $S$ is a convex set, otherwise $S$ is not convex.
    \item If $z\in I^i(S)$ then $z\in H(S)$, otherwise $z\not\in H(S)$.
  \end{itemize}

By the above observations, the problems {\sc Convex Set} and {\sc Convex Hull Determination} can be solved in polynomial time for $(V(G),\C)$.
\end{proof}

Let $S\subseteq V(G)$. If $S$ is not convex then an {\em augmenting set} of $S$ is any set $S'$ such that $S\subset S'\subseteq H(S)$ (where the symbol $\subset$ stands for proper inclusion).

\begin{proposition}
\label{prop:cspo}
Let $(V(G),\C)$ be a convexity space. If there is a polynomial-time certification algorithm to solve {\sc Convex Set} for $(V(G),\C)$ that outputs an augmenting set when the problem has a negative answer then {\sc Convex Hull Determination} can also be solved in polynomial time for $(V(G),\C)$.
\end{proposition}

\begin{proof}

Let $S\subseteq V(G)$ and $z\in V(G)$. Since {\sc Convex Set} can be solved in polynomial time for $(V(G),\C)$, we can apply a convexity test to $S$ in polynomial time. If the test succeeds then $S$ is a convex set and, consequently, the convex hull of $S$ is $S$ itself; otherwise, the test fails and there exists a set $S_1$ (a certificate outputted by the test) that can be used to augment the original set $S$. Recall that $S$ is properly contained in $S_1$.

In order to establish the convex hull of $S$, we are going to apply the convexity test algorithm successively, always obtaining a set $S_{i+1}$ that augments $S_i$, until it cannot be augmented anymore. Let $j$ be the smallest index such that $S_j$ results in a convex set. Observe that $j\leq |V(G)\setminus S|$, since at least one vertex is added in each iteration. Moreover, each convexity test can be done in polynomial time. Thus, the entire process still remains polynomial.

If a vertex $z\in V(G)$ is such that $z\in S_j$ then $z\in H(S)$. Therefore, {\sc Convex Hull Determination} is in $\Py$ for $(V(G),\C)$. \end{proof}

Note that Propositions~\ref{prop:idpo} and~\ref{prop:cspo} can be used to fill some entries of Table~\ref{conv-vs-prob}. For example, since {\sc Interval Determination} is in $\Py$ for the geodesic convexity, by Proposition~\ref{prop:idpo} the problems {\sc Convex Set} and {\sc Convex Hull Determination} are also in $\Py$ for such convexity. The same applies to the $P_3$- and $P^*_3$- convexities. On the other hand, Proposition~\ref{prop:cspo} implies that {\sc Convex Hull Determination} is in $\Py$ for the monophonic convexity.

\section{A general framework for path convexities}\label{frame}

In this section, we propose a general framework for the study of path convexities.

From now on, we assume that every $n$-vertex graph $G$ has vertices labeled $1,2,\ldots,n$. A {\em length matrix} is a symmetric $n\times n$ matrix $M$ such that each entry $M(i,j)$, for $i,j\in V(G)$, is a natural number; in addition, all diagonal entries of $M$ are zero.

Let $A,B,C,D$ be four $n \times n$ length matrices. Suppose that $\P$ is the family of paths of $G$ such that an $ij$-path $P$ of $G$ is a member of $\P$ if and only if:
\begin{enumerate}
\item[(1)] $|P|\geq A(i,j)$;
\item[(2)] $|P|\leq B(i,j)$;
\item[(3)] all the chords in $P$ are of length at least $C(i,j)$;
\item[(4)] all the chords in $P$ are of length at most $D(i,j)$.
\end{enumerate}

Let $I_{\P}:2^{V(G)}\rightarrow 2^{V(G)}$ be the interval function associated with $\P$, and let $\C_{\P}$ be the family of subsets of $V(G)$ such that $S\in\C_{\P}$ if and only if $I_{\P}(S)=S$. Since $\P$ is a particular collection of paths of $G$, by Proposition~\ref{prop:convexity}, we have that $(V(G),\C_{\P})$ is a finite convexity space, equipped with interval function $I_{\P}$. Let us say that such a convexity space defines a {\em matrix path convexity}.

Again, we omit the subscript $\P$ when it is clear from the context.

Say that an $ij$-path $P$ {\em satisfies} matrices $A,B,C,D$ if all the conditions (1) to (4) above are satisfied by $P$.

\subsection{Putting the matrices as part of the input}\label{sec:input-matrices}

In the six problems listed in Section~\ref{problems}, the graph $G$ is always part of the input; however, the rule that determines which collection of paths of $G$ must be considered is {\em not} part of the input. More general versions of such problems are possible when the desired convexity space, expressed as a graph $G$ together with a set of four length matrices, is part of the input. For example, consider the following version of {\sc Convex Set}:

\noindent {\sc Matrix Convex Set}\\
Input: A graph $G$, four $n\times n$ length matrices $A,B,C,D$, and $S\subseteq V(G)$.\\
Question: Is $S$ convex under the matrix path convexity ruled by $A,B,C,D$?

All the remaining problems listed in Section~\ref{problems} can be restated analogously.

The next theorems say that such ``matrix problems'' are all hard. However, we shall see that restrictions on the matrices $A,B,C,D$ lead to interesting cases. In this regard, some types of length matrices are of special interest. For a graph $G$, the {\em distance matrix} of $G$ is the length matrix $M_{\mathit dist}$ with entries $M_{\mathit dist}(i,j)={\mathit dist}_G(i,j)$, for $i,j\in V(G)$. For a positive integer constant $k$, the $(n-k)${\em -matrix} and the $k${\em -matrix} are the length matrices $M_{n-k}$ and $M_k$ with off-diagonal entries all equal to, respectively, $n-k$ and $k$.

\begin{theorem}\label{the:mcs}
{\sc Matrix Convex Set} is Co-NP-complete.
\end{theorem}

\begin{proof}
A certificate for a negative answer to {\sc Matrix Convex Set} is a triple $i,j,z$ (with $i,j\in S$ and $z\not\in S$) and an $ij$-path $P$ in $G$ containing $z$ such that $P$ satisfies $A$ to $D$. Such a certificate can be clearly checked in polynomial time. Therefore, {\sc Matrix Convex Set} is in $\CONP$.

To prove that {\sc Matrix Convex Set} is Co-NP-complete, we show a reduction from the following NP-complete problem~\cite{HS06}: given three distinct vertices $i,j,z$ in a graph $H$, decide whether there is a chordless $ij$-path passing through $z$.

Let $G$ be the graph obtained from $H$ by replacing each edge $(s,z)$ incident to $z$ by an $sz$-path containing $n-1$ internal vertices of degree two, where $n=|V(H)|$. In other words, $G$ is a subdivision of $H$ obtained by subdividing each edge incident to $z$ using $n-1$ vertices. Set $A$ and $B$ as the length matrices with off-diagonal entries all equal to, respectively, $2n$ and $3n-3$. Also, set $C=D=M_1$ (the $k$-matrix for $k=1$). Finally, set $S=\{i,j\}$. Note that the collection of paths $\P$ defined by $A,B,C,D$ contains the chordless paths with length at least $2n$ and at most $3n-3$.

Suppose that there is a chordless $ij$-path $P_H$ in $H$ passing through $z$. Write $P_H=(s_0=i,s_1,\ldots,s_{h-1},s_h=z,s_{h+1},\ldots,s_l=j)$. Then there is a chordless $ij$-path $P_G$ in $G$ obtained from $P_H$ by subdividing edges $(s_{h-1},z)$ and $(z,s_{h+1})$ using $n-1$ vertices of degree two for each edge. Note that $|P_G|=(l-2)+2n$. Since $2\leq l\leq n-1$, we have $2n\leq |P_G|\leq 3n-3$. Therefore, $P_G$ satisfies $A,B,C,D$, and its existence implies that $S$ is not convex.

Conversely, suppose that $S$ is not convex. Then there is a chordless $ij$-path $P_G$ of length at least $2n$ passing through some vertex of $G$ lying outside $S$. But, by the construction of $G$, all the $ij$-paths of length at least $2n$ must necessarily pass through $z$. Let $P_{hz}$ be the subpath of $P_G$ with length $n$ that starts at a vertex $h$ and ends at $z$. Similarly, let $P_{zh'}$ be the subpath of $P_G$ with length $n$ that starts at $z$ and ends at a vertex $h'$. By replacing $P_{hz}$ and $P_{zh'}$ by edges $(h,z)$ and $(z,h')$, we obtain a chordless $ij$-path in $H$ passing through $z$. This completes the proof.
\end{proof}

\begin{theorem}\label{the:mid}
{\sc Matrix Interval Determination} is NP-complete.
\end{theorem}

\begin{proof}
A certificate for a positive answer to {\sc Matrix Interval Determination} is a pair $i,j$ (with $i,j\in S$) and an $ij$-path $P$ in $G$ containing $z$ (recall that $z$ is part of the input) such that $P$ satisfies $A,B,C,D$. Moreover, this certificate can be checked in polynomial time. Therefore, {\sc Matrix Interval Determination} is in $\NP$.

To prove that {\sc Matrix Interval Determination} is NP-complete, recall from Table~\ref{conv-vs-prob} that {\sc Interval Determination} is NP-complete for the monophonic convexity. If $A=M_2$, $B=M_{n-1}$, and $C=D=M_1$, the collection of paths $\P$ associated with such matrices is precisely the collection of induced paths of $G$. Then {\sc Interval Determination} for the monophonic convexity is a restriction of {\sc Matrix Interval Determination}, i.e., the former problem is NP-complete.
\end{proof}

\begin{theorem}\label{the:mchd}
{\sc Matrix Convex Hull Determination} is NP-complete.
\end{theorem}

\begin{proof}
A certificate for a positive answer to {\sc Matrix Convex Hull Determination} is formed by a sequence of paths $P_1,P_2,\ldots,P_r$ such that $r\leq n-1$, $z\in V(P_r)$, and each $P_k$ is an $i_kj_k$-path satisfying matrices $A,B,C,D$ with $i_k,j_k \in S\cup V(P_1)\cup\cdots\cup V(P_{k-1})$. It is easy to see that such a certificate can be checked in polynomial time. Therefore, {\sc Matrix Convex Hull Determination} is in $\NP$.

For the hardness proof, we use the same reduction described in Theorem~\ref{the:mcs}. Again, if there is a chordless $ij$-path $P_H=(i,\ldots,h,z,h',\ldots,j)$ in $H$ then there is a chordless $ij$-path $P_G$ in $G$ obtained from $P_H$ by subdividing edges $(h,z)$ and $(z,h')$, as explained in the proof of Theorem~\ref{the:mcs}, such that $P_G$ satisfies $A,B,C,D$. Therefore, $z\in I(S)\subseteq H(S)$.

Conversely, suppose that $z\in H(S)$. Then there is a sequence of paths $P_1,P_2,\ldots,P_r$ such that $z\in V(P_r)$ and, for $1\leq k\leq r$, $P_k$ is an $i_kj_k$-path satisfying matrices $A,B,C,D$, where $i_k,j_k \in S\cup V(P_1)\cup\cdots\cup V(P_{k-1})$.

Note that $i_1,j_1\in S$. Thus, $\{i_1,j_1\}=\{i,j\}$. In addition, $|P_1|\geq 2n$. But, by the construction of $G$, all the $ij$-paths of length at least $2n$ must necessarily pass through $z$. Hence, $z\in V(P_1)$. The rest of the proof follows as in the proof of Theorem~\ref{the:mcs}: let $P_{hz}$ (resp., $P_{z h'}$) be the subpath of $P_1$ with length $n$ starting at some $h$ (resp., at $z$) and ending at $z$ (resp., at some $h'$). Replacing $P_{hz}$ and $P_{zh'}$ by edges $(h,z)$ and $(z,h')$ produces a chordless $ij$-path in $H$ passing through $z$.
\end{proof}

\begin{theorem}\label{the:mcn-min-mhn} \ \ \ $1.$ {\sc Matrix Convexity Number} is NP-hard.\\
\hspace*{2.9cm} $2.$ {\sc Matrix Interval Number} is NP-complete.\\
\hspace*{2.9cm} $3.$ {\sc Matrix Hull Number} is NP-complete.
\end{theorem}

\begin{proof}
We first prove that both {\sc Matrix Interval Number} and {\sc Matrix Hull Number} are in $\NP$.


A certificate for a positive answer to {\sc Matrix Interval Number} consists of a set $S\subseteq V(G)$ of size at most $r$, and a collection of paths $\{P_z\mid z\in V(G)\setminus S\}$ such that for each path $P_z$ there exist $i_z,j_z\in S$ for which $P_z$ is an $i_zj_z$-path containing $z$ and satisfying $A,B,C,D$. Since the number of paths in the collection is $O(n)$, {\sc Matrix Interval Number} is in $\NP$.

A certificate for a positive answer to {\sc Matrix Hull Number} consists of a set $S$ of size at most $r$ and a sequence of paths $P_1,P_2,\ldots,P_s$ such that:\\
\hspace*{1cm} (a) $s\leq n-1$;\\
\hspace*{1cm} (b) $V(G)\setminus S\subseteq V(P_1)\cup\cdots\cup V(P_s)$;\\
\hspace*{1cm} (c) for $1\leq k\leq s$, $P_k$ is an $i_kj_k$-path satisfying matrices $A,B,C,D$,\\
\hspace*{1.7cm} with $i_k,j_k \in S\cup V(P_1)\cup\cdots\cup V(P_{k-1})$.

Since such a certificate can be checked in polynomial time, {\sc Matrix Hull Number} is in $\NP$.

Now we describe the hardness proof for the three problems. As in Theorem~\ref{the:mid}, we use a proof by restriction. Recall from Table~\ref{conv-vs-prob} that {\sc Convexity Number}, {\sc Interval Number}, and {\sc Hull Number} are all NP-complete for the geodesic convexity. If $A=B=M_{\mathit dist}$ and $C=D=M_1$, the collection of paths $\P$ associated with such matrices is precisely the collection of shortest paths of $G$. This means that {\sc Convexity Number}, {\sc Interval Number}, and {\sc Hull Number} for the geodesic convexity are, respectively, restrictions of {\sc Matrix Convexity Number}, {\sc Matrix Interval Number}, and {\sc Matrix Hull Number}. Hence, the theorem follows.
\end{proof}

\subsection{Constant matrices: the $(a,b,c,d)$-path convexity}\label{sec:constant}

In this section, we study the case in which there are constants $a,b,c,d$ such that $A=M_a$, $B=M_b$, $C=M_c$, and  $D=M_d$. In this scenario we can assume that the matrices are {\em not} part of the input, because length restrictions are known in advance. This gives rise to ``constant matrix versions'' of the problems studied in the preceding subsection. For example, consider the following problems:

\noindent {\sc $(a,b,c,d)$-Convex Set}\\
Input: A graph $G$ and a set $S\subseteq V(G)$.\\
Question: Is $S$ convex under the matrix path convexity ruled by $A,B,C,D$?\\
{\em Equivalently:} Is $S$ convex under the path convexity defined by the collection $\P(a,b,c,d)$ of paths of $G$ whose length is at least $a$ and at most $b$, and whose chords have length at least $c$ and at most $d$?

\noindent {\sc $(a,b,c,d)$-Interval Determination}\\
Input: A graph $G$, a set $S\subseteq V(G)$, and a vertex $z\in V(G)$.\\
Question: Does $z$ belong to $I(S)$, where $I$ is the interval function associated with the collection $\P(a,b,c,d)$ of paths of $G$?

The remaining matrix problems can be restated analogously.

The path convexity for which the path/chord length restrictions are ruled by four constants $a,b,c,d$ as explained above is called {\em $(a,b,c,d)$-path convexity}.

\begin{theorem}\label{the:cmid}
{\sc $(a,b,c,d)$-Interval Determination} is in $\Py$.
\end{theorem}

\begin{proof}
Note that for a constant $\ell$ there are $O(n^{\ell+1})$ paths of length $\ell$ in $G$. Thus there are $O(\sum_{\ell=a}^{b} n^{\ell+1})$ paths in $G$ with length at least $a$ and at most $b$. Now, for each pair $(i,j)$ of distinct vertices in $S$, there are $O(\sum_{\ell=a}^{b} n^{\ell-1})$ $ij$-paths with length at least $a$ and at most $b$, because $i$ and $j$ are the fixed endpoints of each such path. But $S$ contains $O(|S|^2)$ pairs of distinct vertices. This amounts to checking $O(|S|^2 \sum_{\ell=a}^{b} n^{\ell-1})$ paths. Since a path of length $\ell$ can have at most $O(\ell^2)$ chords, we can select all the $ij$-paths in $\P(a,b,c,d)$ with $i,j\in S$ in $O(|S|^2 \sum_{\ell=a}^{b} \ell^2 n^{\ell-1})$ time. Finally, checking whether $z$ belongs to one of such paths can be done in $O(1)$ time per path, because the lenght of each path is bounded by $b$. This gives a na\"ive polynomial-time brute-force algorithm to check whether $z\in I(S)$.
\end{proof}

By Proposition~\ref{prop:idpo}, we have:

\begin{corollary}
{\sc $(a,b,c,d)$-Convex Set} and {\sc $(a,b,c,d)$-Convex Hull Determination} are in $\Py$. \ \ $\blacksquare$
\end{corollary}

As for the other three problems, {\sc $(a,b,c,d)$-Convexity/Interval/Hull Number}, we remark that the special cases
$$(a=2,b=2,c=1,d=2) \mbox{ and } (a=2,b=2,c=1,d=1)$$
correspond precisely to the $P_3$- and $P^*_3$- convexities, as indicated in Table~\ref{convexities}. For both convexities, all the three problems are NP-complete (see Table~\ref{conv-vs-prob}).

\subsection{$(a,b,c,d)$-path convexity and bounded treewidth graphs}

In this section, we investigate the complexity of the six $(a,b,c,d)$-path convexity problems in Section~\ref{sec:constant} when applied to bounded treewidth graphs. As we shall see, linear-time methods will be possible in this case.

Let $G$ be a graph, $T$ a tree, and $V=(V_t)_{t \in T}$ a family of vertex sets $V_t\subseteq V(G)$ indexed by the vertices $t$ of $T$. The pair $(T,V)$ is called a {\em tree-decomposition of $G$} if it satisfies the following three conditions~\cite{diestel}:
\begin{description}
\item [(T1)] $V(G)=\bigcup_{t\in T} V_t$;
\item [(T2)] for every edge $e\in G$ there exists $t\in T$ such that both ends of $e$ lie in $V_t$;
\item [(T3)] if $V_{t_i}$ and $V_{t_j}$ both contain a vertex $v$ then $v\in V_{t_k}$ for all vertices $t_k$ in the path between $t_i$ and $t_j$.
\end{description}

The {\em width} of $(T,V)$ is the number $\max\{|V_t|-1 \mid t\in T\}$, and the {\em treewidth} $tw(G)$ of $G$ is the minimum width of any tree-decomposition of $G$.

Graphs of treewidth at most $k$ are called {\em partial $k$-trees}. Some graph classes with bounded treewidth include: forests (treewidth 1); pseudoforests, cacti, series-parallel graphs, and outerplanar graphs (treewidth at most 2); Halin graphs and Apollonian networks (treewidth at most 3)~\cite{B98,Br99}. Control flow graphs arising in the compilation of structured programs also have bounded treewidth (at most 6)~\cite{T98}.


In 1990, Courcelle~\cite{C90} stated that for any graph $G$ with treewidth bounded by a constant $k$ and for any graph property $\Pi$ that can be formulated in CMSOL$_2$ ({\em Counting Monadic Second-Order Logic} where quantification over sets of vertices or edges and predicates testing the size of sets modulo constants are allowed), there is a linear-time algorithm that decides if $G$ satisfies $\Pi$ ~\cite{C90,C93,C97,Co11}. This result has been extended a number of times~\cite{Ar91,Bo92,Co00,Hl08}. In particular, Arnborg and Lagergren~\cite{Ar91} study optimization problems over sets definable in Counting Monadic Second-Order Logic.


By Courcelle's meta-theorems based on CMSOL$_2$~\cite{C90,C93,C97}, obtaining linear-time methods to solve the six problems of Section~\ref{sec:constant} on bounded treewidth graphs   amounts to showing that the related properties are expressible in CMSOL$_2$.


\begin{theorem}\label{the:mid_tw}
{\sc $(a,b,c,d)$-Interval Determination} is solvable in linear time on bounded treewidth graphs.
\end{theorem}

\begin{proof}
It is enough to show that the property ``$z\in I(S)$'' is CMSOL$_2$-expressible. Given $G$, $S$, and $z$, we construct $\varphi(G,S,z,a,b,c,d)$ such that $z \in I(S) \Leftrightarrow \varphi(G,S,z,a,b,c,d)$ as follows:

\begin{equation}\label{IDCMSOL1}
\begin{split}
(~z \in S~)~{\vee}\\
(~\exists~ u,v,P \mbox{ } & (~u,v \in S~{\wedge}\\
                                                            & ~~P \mbox{ is an }uv\mbox{-path} ~{\wedge} \\
                                                            & ~~z \mbox{ is in } P~{\wedge} \\
                                                            & ~~Card(P) \geq a~{\wedge}\\
                                                            & ~~Card(P) \leq b~{\wedge} \\
                                                            & ~~\forall P'\mbox{ }( \ (~P'\subseteq P~\wedge\\
                                                            & \mbox{~~~~~~~~~~}Card(P')\geq 2~\wedge \\
                                                            & \mbox{~~~~~~~~~~}\exists~ u',v'
                                                             (P' \mbox{ is an $u'v'$-path}~\wedge~adj(u',v')))\\
                                                             & \mbox{~~~~~~~~}~\Rightarrow
                                                            (Card(P')\geq c \ \ {\wedge} \ \  Card(P')\leq d) \ )\\
                                                            & ))
\end{split}
\end{equation}

In the above formula, paths are regarded as subsets of edges. Using this approach, the  subformula``$P$ is an $uv$-path'' can be expressed in CMSOL$_2$ (see~\cite{C97}). Note that a chord is expressed as an $u'v'$-subpath $P'$ of $P$ with length at least $c$ and at most $d$ such that $u'$ is adjacent to $v'$.
\end{proof}

\begin{corollary}~\label{cor:mcs_tw}
{\sc $(a,b,c,d)$-Convex Set} can be solved in linear time on bounded treewidth graphs.
\end{corollary}
\begin{proof}
The property ``$S$ is convex'' is equivalent to ``there is no $z$ such that $z\not\in S$ and $z\in I(S)$''. By Theorem~\ref{the:mid_tw}, ``$z\in I(S)$'' is CMSOL$_2$-expressible. Thus the result easily follows.
\end{proof}

\begin{corollary}
{\sc $(a,b,c,d)$-Convex Hull Determination} can be solved in linear time on bounded treewidth graphs.
\end{corollary}
\begin{proof}
The property ``$z\in H(S)$'' is equivalent to ``there exists $S_1$ such that: (a) $S_1$ is convex, (b) $S\subseteq S_1$, (c) $z\in S_1$, and (d) there is no $S_2$ such that $S_2$ is convex, $S\subseteq S_2$, and $S_2$ is properly contained in $S_1$''. By Corollary~\ref{cor:mcs_tw}, we can use CMSOL$_2$ to say that the sets $S_1$ and $S_2$ are convex. Thus the result follows.
\end{proof}

For the remaining three problems ({\sc $(a,b,c,d)$-Convexity/Interval/Hull Number}), we consider their optimization versions (maximization in the case of {\sc Convexity Number}, and minimization in the case of {\sc Interval/Hull Number}). Note that the properties ``$S$ is a convex set distinct from $V(G)$'', ``$S$ is an interval set'', and ``$S$ is a hull set'' can be expressed in CMSOL$_2$. Therefore the optimization versions (``find an optimal set satisfying the required property'') are LinCMSOL$_2$ problems~\cite{Ar91,C92}, to which the following result applies:

\begin{theorem} {\em \cite{Ar91,C92}}
Let $k$ be a positive constant, and $\Pi$ be a {\em LinCMSOL}$_2$ problem. Then $\Pi$ can be solved in linear time on graphs of treewidth bounded by $k$ (if the tree-decomposition is given with the input graph).
\end{theorem}

Therefore:

\begin{corollary}
The optimization versions of {\sc $(a,b,c,d)$-Convexity Number}, {\sc $(a,b,c,d)$-Interval Number}, and {\sc $(a,b,c,d)$-Hull Number} can be solved in linear time on bounded treewidth graphs (if the tree-decomposition is given with the input graph).
\end{corollary}

\subsection{Particular cases of the $(a,b,c,d)$-path convexity}\label{sec:partcases}

In this section we show that, by extending the meaning of the parameters $a,b,c,d$, some path convexities in the literature can be viewed as particular cases of the $(a,b,c,d)$-path convexity. In Table~\ref{abcd-path-convexities} below, the symbol `$\sigma$' (resp.,`$\ell$') means that the length of the shortest (resp., longest) path between each pair of distinct vertices must be considered. The symbol `$\infty$' stands for no length restriction. For a constant $k$, the symbol `$k\mid\sigma$' means that, for each pair $(i,j)$ of distinct vertices, the minimum value between $k$ and the length of the shortest $ij$-path must be considered.

\begin{table}[H]
\centering
\small
\begin{tabular}{ | l | C{2.0cm} | C{2.7cm} | C{2.0cm} | C{2.0cm} | }
\toprule
{\bf Convexity}&\textit{\textbf{a}}&\textit{\textbf{b}}&\textit{\textbf{c}}&\textit{\textbf{d}}\\ \hline
geodesic        & $\sigma$           & $\sigma$                 & $1$             & $1$      \\
monophonic      & $2$                & $\infty$                 & $1$             & $1$      \\
$g^3$           & $3$                & $\sigma$                 & $1$             & $1$      \\
$g_k$           & $\sigma$           & $k\mid\sigma$            & $1$             & $1$      \\
$m^3$           & $3$                & $\infty$        			& $1$             & $1$      \\
$P_3$           & $2$                & $2$            			& $1$             & $2$      \\
$P_3^*$         & $2$                & $2$            			& $1$             & $1$      \\
triangle-path   & $2$                & $\infty$        			& $1$             & $2$      \\
total           & $2$                & $\infty$        			& $3$             & $\infty$ \\
detour          & $\ell$             & $\ell$   	            & $1$             & $\infty$ \\
all-path        & $2$                & $\infty$        			& $1$             & $\infty$ \\
\bottomrule
\end{tabular}
\caption{Path convexities as particular cases of the $(a,b,c,d)$-path convexity. Note that putting $c=d=1$ implies that all the paths of the considered collection $\P$ are chordless.\label{abcd-path-convexities}}
\end{table}

\section{Concluding remarks}

In this work we described a matrix path convexity framework, where, by means of four input matrices, we can specify the types of paths that must be considered for each pair of vertices of the input graph, ``customizing'' the path convexity to be dealt with. Since this general approach results in the hardness of the related computational problems (at least the more studied ones), we also investigate the case of constant matrices. The latter case leads to the study of the $(a,b,c,d)$-path convexity, where the rule that defines the convexity is not part of the input. If $a$, $b$, $c$, and $d$ are positive constants then, in such convexity, the problems of computing the interval of a set, deciding whether a set is convex, and computing the convex hull of a set, are all solvable in polynomial time. In addition, all the ``$(a,b,c,d)$-versions'' of the six problems listed in Section~\ref{problems} are solvable in linear time if the input graph has bounded treewidth. We have also shown that, by extending the meaning of the parameters $a,b,c,d$, most path convexities considered in the literature can be viewed as particular cases of the $(a,b,c,d)$-path convexity. In this regard, other interesting convexities, not yet considered in the literature up to the authors' knowledge, can be defined by choosing other values for the tuple $(a,b,c,d)$.  Tables~\ref{convexitiesproposed} and~\ref{abcd-path-convexities-proposed} describe such convexities. The symbol `n$^-$\;' means that paths with length one less than the number of vertices of the input graph must be considered.

\begin{table}[htb]
\centering
\small
\begin{tabular}{ |l|c|c|c|c|c| }
\toprule
convexity name & collection of paths $\P$ considered\\ \hline
$g^k$          & shortest paths of length at least $k$ \\
$m^k$          & induced paths of length at least $k$ \\
$(k,l)$-path   & chordless paths of length between $k$ and $l$ \\
$k$-path       & chordless paths with $k$ vertices\\
Hamiltonian    & Hamiltonian paths \\
\bottomrule
\end{tabular}
\caption{Some new convexities proposed in this work.}\label{convexitiesproposed}
\end{table}

\begin{table}[htb]
\centering
\small
\begin{tabular}{ | l | C{1.5cm} | C{1.5cm} | C{1.5cm} | C{1.5cm} | }
\toprule
{\bf Convexity} &\textit{\textbf{a}} &\textit{\textbf{b}} &\textit{\textbf{c}} &\textit{\textbf{d}}\\ \hline
$g^k$               & $k$     & $\sigma$     & $1$     & $1$      \\
$m^k$               & $k$     & $\infty$     & $1$     & $1$      \\
$(k,l)$-path        & $k$     & $l$          & $1$     & $1$      \\
$k$-path            & $k-1$   & $k-1$        & $1$     & $1$    \\
Hamiltonian         & n$^-$   & n$^-$        & $1$     & $\infty$  \\
\bottomrule
\end{tabular}
\caption{Convexities from Table~\ref{convexitiesproposed} described according the $(a,b,c,d)$-path convexity framework.\label{abcd-path-convexities-proposed}}
\end{table}

Let $\mathds{K}=\{k \mid \sigma\}_{k\in\mathds{N}^*}$ and $\Sigma=\mathds{N}^*\cup\{\sigma,\infty,\ell,\mathrm{n}^-\}\cup\mathds{K}$, and consider the domain of tuples $(a,b,c,d)\in\Sigma^4$ (considering, of course, only meaningful cases). Such domain can be used to systematize algorithmic studies in path convexity in some ways. One example is to find complexity dichotomies (complete classifications of the complexity of a fixed computational problem $\Pi$). For example, if $\Pi=$\,{\sc Interval Determination}, $c=1$, and $d\in\{1,2\}$, which values of $(a,b,c,d)$ imply tractability (hardness) of $\Pi$ under the $(a,b,c,d)$-path convexity? From Tables~\ref{conv-vs-prob} and~\ref{abcd-path-convexities} we know that the cases $(\sigma,\sigma,1,1)$, $(2,2,1,1)$, and $(2,2,1,2)$ are in $P$, while $(2,\infty,1,1)$ and $(2,\infty,1,2)$ are NP-complete.


\end{document}